\documentclass[a4paper,12pt]{article}
\pdfoutput=1
\usepackage{test}

\usepackage{xcolor}
\definecolor{pantone}{RGB}{1, 33, 105}

\usepackage{hyperref}
\hypersetup{colorlinks,linkcolor=blue,citecolor=blue,urlcolor=pantone}
\usepackage{doi}

\usepackage[style=ext-authoryear-comp,
sorting=nyt, 
sortcites=false; 
dashed=false, 
maxcitenames=2, 
maxbibnames=99, 
uniquelist=false,
uniquename=false,
giveninits=true, 
natbib, 
date=year 
]{biblatex}

\AtBeginRefsection{\GenRefcontextData{sorting=ynt}}
\AtEveryCite{\localrefcontext[sorting=ynt]}

\DeclareFieldFormat{pages}{#1} 
\renewbibmacro{in:}{\ifentrytype{article}{}{\printtext{\bibstring{in}\intitlepunct}}} 

\DeclareFieldFormat[article,inbook,incollection,inproceedings,patent,thesis,unpublished]{titlecase:title}{\MakeSentenceCase*{#1}} 

\newbibmacro*{atoda:titlelink}[1]{%
  \ifhyperref
    {\iffieldundef{doi}
       {\iffieldundef{eprint}
          {\iffieldundef{url}
             {#1}
             {\href{\thefield{url}}{#1}}}
          {\href{https://arxiv.org/abs/\thefield{eprint}}{#1}}}
       {\href{https://doi.org/\thefield{doi}}{#1}}}
    {#1}}

\DeclareFieldFormat{title}{\usebibmacro{atoda:titlelink}{\mkbibemph{#1}}}
\DeclareFieldFormat
  [article,inbook,incollection,inproceedings,patent,thesis,unpublished]
  {title}{\usebibmacro{atoda:titlelink}{\mkbibquote{#1\isdot}}}

\renewbibmacro*{doi+eprint+url}{}

\usepackage[inline,shortlabels]{enumitem}
\setlist[enumerate,1]{label=(\roman*)}

\usepackage[margin = 1.25in]{geometry}
\usepackage{setspace}
\newcommand{\AAT}[1]{{\color{magenta} #1}}

\title{Equilibrium Selection in Pure Bubble Models by Dividend Injection}
\author{Tomohiro Hirano\thanks{Department of Economics, Royal Holloway, University of London  \href{mailto:tomohiro.hirano@rhul.ac.uk}{tomohiro.hirano@rhul.ac.uk}.} \and Alexis Akira Toda\thanks{Department of Economics, Emory University, \href{mailto:alexis.akira.toda@emory.edu}{alexis.akira.toda@emory.edu}.}}
\date{Accepted at \emph{Economic Theory}}

\numberwithin{equation}{section}
\numberwithin{prop}{section}

\begin{document}

\maketitle

\begin{abstract}

Rational ``pure bubble'' models admit multiple (often a continuum of) equilibria, making predictions and policy analysis fragile. We show that if the fundamental equilibrium features a low interest rate relative to growth ($R<G$), then a bubbly equilibrium with $R=G$ exists. We propose an equilibrium selection device based on \emph{dividend injection}: endowing the bubble asset with a small dividend stream that grows more slowly than the aggregate economy. This perturbation eliminates the fundamental steady state and all equilibrium paths converging to it, thereby resolving indeterminacy. We illustrate the general applicability of dividend injection in overlapping-generations and infinite-horizon models, with or without production and financial frictions. In particular, we argue that the widely applied ``stochastic bubble'' model of \citet{Blanchard1979} and \citet{Weil1987} is implausible.

\medskip

\noindent
\textbf{Keywords:} bubble, dividend, equilibrium indeterminacy, growth, low interest rate, necessity.

\medskip

\noindent
\textbf{JEL codes:} D53, E44, G12.
\end{abstract}

\section{Introduction}

An asset price bubble is a situation in which the asset price ($P$) exceeds its fundamental value ($V$), which is defined as the present discounted value of dividends ($D$). In the so-called ``rational bubble'' model, the asset price exceeds its fundamental value as an equilibrium outcome, even if agents hold rational expectations and common knowledge about the asset.\footnote{\label{fn:RB}The idea of the ``rational bubble'' model is to start with the no-arbitrage asset pricing equation, iterate it forward, and define the fundamental value of the asset as the present discounted value of dividends, where discounting is done using the stochastic discount factor (SDF) of the \emph{marginal} agent; see \citet{HiranoToda2025EJW} for more details and \citet{MartinVentura2018} and \citet{HiranoToda2024JME} for reviews of this literature. Other common approaches to studying asset price bubbles define the fundamental value using the SDF of \emph{particular} agents and declare that the asset exhibits a bubble if its price exceeds these valuations. Examples of these approaches include heterogeneous beliefs \citep{HarrisonKreps1978,ScheinkmanXiong2003,FostelGeanakoplos2012Tranching} and asymmetric information \citep{AllenMorrisPostlewaite1993,AbreuBrunnermeier2003,Conlon2004,Doblas-Madrid2012,Doblas-Madrid2016,LiuConlon2018,AllenBarlevyGale2022,AwayaIwasakiWatanabe2022,AraujoDoblas-Madrid2026,Liu2026}. 
} The literature on rational asset price bubbles has almost exclusively focused on a special case called ``pure bubbles'', namely assets that pay no dividends ($D=0$) and are intrinsically worthless, like fiat money. Because the fundamental value of a pure bubble asset is $V=0$, we say that an asset price bubble exists whenever $P>0$ in equilibrium. Although pure bubble models capture important aspects of asset price bubbles, one of their shortcomings is that they often exhibit multiple equilibria.\footnote{See \citet[\S4.7]{HiranoToda2024JME} for a more extensive discussion of shortcomings of pure bubble models.} In these models, there often exist a fundamental equilibrium in which the price of the pure bubble asset is zero, a bubbly steady state in which the asset price is a positive constant, and a continuum of bubbly equilibria converging to the fundamental steady state. Therefore, it is not obvious which equilibrium we should select.

Although the equilibrium indeterminacy in pure bubble models has been recognized for decades \citep{Gale1973},\footnote{In a typical pure bubble model such as \citet{Tirole1985}, only when the initial bubble price just equals the price that corresponds to a saddle path (which corresponds to the largest sustainable bubble size), the economy will converge to a steady state with positive bubbles. Otherwise, the bubble price will converge to zero, and there is a continuum of such asymptotically bubbleless paths. See \citet{HiranoToda2024EL} as well as \S\ref{sec:example} below for a formal analysis.} the literature has selected only one of a continuum of bubbly equilibria (a saddle path or a steady state) and has advanced policy and quantitative analyses. However, the selected equilibrium is only one element of a continuum, and the model itself supplies no criterion for selecting it: there is no reason why heterogeneous economic agents would coordinate on that equilibrium, even if they are exposed to disturbances, including policy changes. In short, equilibrium indeterminacy makes model implications and predictions fragile and non-robust.

In this paper, we propose a solution to this indeterminacy problem in rational bubble models. We resolve equilibrium indeterminacy by injecting small, positive dividends into the bubble asset. Using a reduced-form general equilibrium model, we show that when the interest rate in the fundamental equilibrium is below the economic growth rate ($R<G$), a bubbly equilibrium with $R=G$ exists. In this case, we can eliminate the fundamental steady state and resolve indeterminacy by introducing a dividend-paying asset with a dividend growth rate below the economy's growth rate ($G_d<G$) but above the interest rate ($G_d>R$). The idea is based on the bubble necessity result of \citet{HiranoToda2025JPE}, which goes as follows. When $G_d<G$, the dividends from the asset are asymptotically negligible, so the asset behaves like a pure bubble asset. If the asset price reflects fundamentals, its price must grow at the same rate as dividends, $G_d<G$, and hence must be negligible relative to the economy in the long run. But the condition $G_d>R$ implies that the fundamental value (the present value of dividends) diverges, which is impossible in equilibrium. Thus, asset price bubbles are necessary for equilibrium to exist by raising the interest rate above the dividend growth rate. Our dividend injection procedure not only shows the necessity of asset price bubbles but also eliminates all equilibria converging to the fundamental steady state. Furthermore, we show that under plausible conditions on the elasticity of substitution (between goods or production factors), there exists a locally determinate bubbly equilibrium even with slightly positive dividends. These results imply that there is a discontinuity in equilibrium determinacy between the case with $D=0$ and the case with $D>0$.

Our equilibrium selection via dividend injection draws on the idea of ``commodity money refinement'' in the monetary theory literature, but there are conceptual differences. For simplicity, consider a model with no growth. The idea of commodity money refinement is to add small dividends to money (the pure bubble asset), study the equilibrium, and take the limit as dividends tend to zero, \ie, the pure-bubble limit. However, there is a conceptual problem: by adding positive, constant dividends to the pure bubble asset, as in commodity money refinement in monetary theory, the asset price becomes \emph{equal} to its fundamental value and hence it is not a bubble. Only in the pure-bubble limit, when dividends vanish, does the asset become a bubble. In other words, there is a discontinuity in the qualitative properties of asset prices at the limit. In contrast, we suggest adding dividends that grow at a slower rate than the economy ($G_d<G$, with $G=1$ in this example). In this case, applying the concept of necessity of asset price bubbles established in \citet{HiranoToda2025JPE}, we show that the asset price contains a bubble (exceeds its fundamental value) whether or not we take the limit. Although the equilibrium ultimately selected is the same under commodity money refinement and dividend injection, our approach may be conceptually appealing because it resolves the discontinuity in the qualitative properties of asset prices. Furthermore, our approach is suited for applications because, when thinking about real assets yielding dividends such as stocks, land, and housing, analyzing the case without taking the pure-bubble limit opens up a wide range of applications, including empirical work.

Our equilibrium selection via dividend injection provides the existing literature on pure bubbles with a firmer footing. However, there are many substantial benefits in establishing the necessity and local uniqueness of the bubbly equilibrium beyond mere theoretical satisfaction.
\begin{enumerate*}
    \item First, it enables comparative dynamics, making model predictions robust for policy and quantitative analyses. For instance, in the applied theory literature, it is common to consider the so-called ``stochastic bubble'' model in which the pure bubble asset permanently loses value with some probability. In \S\ref{subsec:application_stochastic} we show that stochastic bubbles are ruled out by our equilibrium selection, which implies that stochastic bubble models are implausible.
    \item Second, unlike pure bubble assets with no dividends, our idea of dividend injection is also a model of bubbles attached to assets yielding positive dividends. Hence, this approach would be more natural for considering realistic bubbles attached to land, housing, and stocks.
    \item Third, our approach has the potential for extensions and applications. In game theory, \citet{CarlssonVanDamme1993} show that adding a small amount of noise to $2\times 2$ games with incomplete information removes equilibrium multiplicity. \citet{MorrisShin1998} apply this idea to select a unique equilibrium in a model of currency attacks with a continuum of traders. The selection of a unique equilibrium outcome by adding small amounts of noise has opened up enormous possibilities for applications  \citep{MorrisShin2000,MorrisShin2003}. From this perspective, we can expect our approach to produce various applications in macroeconomic analysis with asset price bubbles. For instance, \citet*{BernankeGertlerGilchrist1999} developed the quantitative macroeconomic framework called ``BGG'', which pioneered the large literature studying the link between asset prices and the real economy. Our approach can be embedded into their framework or other quantitative macroeconomic models so that we can conduct various policy and quantitative analyses in a way that asset prices contain a bubble. This direction is of considerable importance because policymakers want to understand leaning against the bubble policy \citep{Barlevy2018}. Our approach provides a theoretical foundation for the direction.\footnote{Note that \citet{BernankeGertler1999} extended \citet{BernankeGertlerGilchrist1999} by adding exogenous asset price bubbles. In contrast, in our approach, asset bubbles occur as the equilibrium outcome within a general equilibrium framework.}
\end{enumerate*}

We acknowledge that dividend injection is not a universal equilibrium selection device. Its applicability requires a low-interest fundamental equilibrium and sufficient regularity for a bubbly balanced-growth equilibrium to exist. Moreover, our determinacy results are local: the procedure need not eliminate cycles, chaotic equilibria, or multiple bubbly steady states.

The rest of the paper is organized as follows. After discussing the literature, \S\ref{sec:example} presents a motivating example to select a unique equilibrium in a canonical pure bubble model. \S\ref{sec:inject} discusses the elimination of fundamental equilibria by dividend injection in a general setting and applies it to overlapping generations and infinite-horizon models with or without production. To illustrate the role of financial frictions, \S\ref{sec:lev} presents an application to an entrepreneurial economy with leverage. \S\ref{sec:application} provides further applications to typical pure bubble models in the literature. Most proofs are relegated to Appendix \ref{sec:proof}.

\subsection{Related literature}

Our idea of introducing small positive dividends to select an equilibrium is conceptually analogous to the equilibrium refinement literature in game theory (see \citet{vanDamme1987} for a review) and monetary theory. Regarding game theory, in both the perfect equilibrium of \citet{Selten1975} and the proper equilibrium of \citet{Myerson1978}, players assign infinitesimally small but positive probabilities to all pure strategies that are not best responses. In monetary theory, to rule out hyperinflationary equilibria in OLG models with money, \citet{Scheinkman1978} and \citet{BrockScheinkman1980} discuss the possibility of introducing small dividends, mandatory savings (social security), or utility from holding money. This approach of introducing small dividends is nowadays commonly referred to as ``commodity money refinement'' and is widely applied.\footnote{Examples include \citet{Wallace1981}, \citet{ObstfeldRogoff1983}, \citet{Nicolini1996}, \citet{WallaceZhu2004}, and \citet{ChoiRocheteau2021}, among others. We thank Guillaume Rocheteau, Pierre-Olivier Weill, Randall Wright, and Tao Zhu for teaching us the history of commodity money refinement.} Furthermore, \citet{AraujoGuimaraes2014,AraujoGuimaraes2017} use the idea of global games to select a unique equilibrium in a monetary model, and \citet{Iwasaki2026} proposes an equilibrium refinement in the spirit of forward induction in a model of cryptocurrencies.

Our paper belongs to the classical rational bubble literature pioneered by \citet{Samuelson1958}, \citet{Bewley1980}, \citet{Tirole1985}, \citet{Kocherlakota1992}, and \citet{SantosWoodford1997}, which studies bubbles as speculation backed by nothing (Footnote \ref{fn:RB}). As discussed in the recent review article of \citet{HiranoToda2024JME}, the rational bubble literature almost exclusively focuses on bubbles attached to intrinsically worthless assets ($D=0$). These are precisely the models that suffer from equilibrium indeterminacy. To the best of our knowledge, \citet[\S7]{Wilson1981} is the first paper that establishes the existence of rational bubbles in dividend-paying assets ($D>0$).\footnote{See \citet*[\S6.1.2]{LeVanPham2016}, \citet*[Example 2]{BosiHa-HuyLeVanPhamPham2018}, and \citet*[Proposition 7]{BosiLeVanPham2022} for other examples.} \citet[\S6]{HiranoToda2024JME} and \citet{HiranoToda2025JPE,HiranoTodaHousingbubble,HiranoToda2025PNAS} generalize \citet{Wilson1981}'s result in workhorse macroeconomic models with stocks, housing, or land. They derive a conceptually new idea of the necessity of asset price bubbles and a new insight that asset pricing implications under balanced growth and unbalanced growth are markedly different. \citet{Tirole1985} recognized the possibility of the nonexistence of fundamental equilibria. However, some of his claims and proofs were incorrect, as discussed in \citet{PhamToda2026ECMA}.

\section{Motivating example}\label{sec:example}

This section presents a motivating example to select a unique equilibrium in a canonical pure bubble model. The analysis presented here is an extension of an example given in \citet[\S3.1]{HiranoToda2024JME}.

We consider a simple overlapping generations (OLG) model with an intrinsically worthless asset as in \citet{Samuelson1958}. Time is discrete and indexed by $t=0,1,\dotsc$. Agents born at time $t$ live for two periods and have utility function
\begin{equation}
    \log c_t^y+\beta \log c_{t+1}^o, \label{eq:utility}
\end{equation}
where $\beta>0$ governs time preference and $c_t^y,c_{t+1}^o$ denote the consumption when young and old. The initial old care only about their consumption. The time $t$ endowments of the young and old are $(e_t^y,e_t^o)=(aG^t,bG^t)$, where $G>0$ is the gross growth rate of the economy and $a,b>0$. There is a unit supply of an intrinsically worthless asset (pure bubble asset) with no dividends, which is initially held by the old.

Let us find all the perfect foresight equilibria of this economy. Let $P_t\ge 0$ be the price of the asset in units of time $t$ consumption and $R_t>0$ be the gross risk-free rate between time $t$ and $t+1$. Since agents live for only two periods and the asset is initially held by old agents, in equilibrium, it is clear that the old sell the entire asset to the young. Therefore, the time $t$ budget constraints imply
\begin{align*}
    &\text{Old:} && c_t^o+P_t\cdot 0=P_t\cdot 1 + e_t^o\iff c_t^o=bG^t+P_t,\\
    &\text{Young:} && c_t^y+P_t\cdot 1=P_t\cdot 0 + e_t^y\iff c_t^y=aG^t-P_t.
\end{align*}
To support this allocation as an equilibrium, it remains to verify the Euler equation (first-order condition) of the young, which is
\begin{equation}
    1=\beta R_t(c_{t+1}^o/c_t^y)^{-1}=\beta R_t\frac{aG^t-P_t}{bG^{t+1}+P_{t+1}}. \label{eq:euler_young}
\end{equation}
If the asset price equals its fundamental value, then $P_t=0$, the equilibrium allocation is autarkic, and \eqref{eq:euler_young} implies the interest rate
\begin{equation}
    R_t=R=\frac{bG}{\beta a}. \label{eq:R_autarky}
\end{equation}
Next suppose that there is an asset price bubble, so $P_t>0$. Then the absence of arbitrage implies $R_t=P_{t+1}/P_t$, so \eqref{eq:euler_young} implies
\begin{equation}
    1=\beta \frac{P_{t+1}}{P_t}\frac{aG^t-P_t}{bG^{t+1}+P_{t+1}}\iff P_{t+1}=\frac{bG^{t+1}P_t}{\beta aG^t-(1+\beta)P_t}. \label{eq:Pt_diff}
\end{equation}
To solve the difference equation \eqref{eq:Pt_diff}, let us introduce the detrended asset price $p_t\coloneqq G^{-t}P_t$. Then \eqref{eq:Pt_diff} becomes
\begin{equation}
    p_{t+1}=\frac{bp_t}{\beta a-(1+\beta)p_t}\iff \frac{1}{p_{t+1}}=\frac{\beta a}{b}\frac{1}{p_t}-\frac{1+\beta}{b}. \label{eq:xt_diff}
\end{equation}
Since \eqref{eq:xt_diff} is a linear difference equation in $1/p_t$, it is straightforward to solve. The general solution is
\begin{equation}
    \frac{1}{p_t}=\begin{cases*}
        \left(\frac{\beta a}{b}\right)^t\left(\frac{1}{p_0}-\frac{1+\beta}{\beta a-b}\right)+\frac{1+\beta}{\beta a-b} & if $\beta a\neq b$,\\
        \frac{1}{p_0}-\frac{1+\beta}{b}t & if $\beta a=b$.
    \end{cases*}\label{eq:pt}
\end{equation}
Since $P_t>0$ if and only if $p_t>0$, a solution with $P_t>0$ for all $t$ is possible only if $\beta a>b$. Under this condition, noting that $p_0=P_0$, the general solution is
\begin{equation}
    P_t=\frac{G^t}{\left(\frac{\beta a}{b}\right)^t\left(\frac{1}{P_0}-\frac{1+\beta}{\beta a-b}\right)+\frac{1+\beta}{\beta a-b}},\label{eq:Pt}
\end{equation}
where $0<P_0\le \frac{\beta a-b}{1+\beta}$ is arbitrary so that $P_t\ge 0$ for all $t$. Using \eqref{eq:Pt}, the consumption of the young is
\begin{equation*}
    c_t^y=aG^t-P_t=G^t\frac{a\left(\frac{1}{P_0}-\frac{1+\beta}{\beta a-b}\right)+\frac{a+b}{\beta a-b}\left(\frac{b}{\beta a}\right)^t}{\frac{1}{P_0}-\frac{1+\beta}{\beta a-b}+\frac{1+\beta}{\beta a-b}\left(\frac{b}{\beta a}\right)^t}>0,
\end{equation*}
so consumption is interior and we indeed have an equilibrium. We can summarize the above derivations in the following proposition.

\begin{prop}\label{prop:example_eq}
The following statements are true.
\begin{enumerate}
\item If $\beta a\le b$, the unique equilibrium is $P_t=0$ and $(c_t^y,c_t^o)=(aG^t,bG^t)$.
\item If $\beta a>b$, there is a continuum of equilibria parametrized by $0\le P_0\le \frac{\beta a-b}{1+\beta}$, where $P_t$ is given by \eqref{eq:Pt} and $(c_t^y,c_t^o)=(aG^t-P_t,bG^t+P_t)$.\footnote{We always use the convention $1/0=\infty$ and $1/\infty=0$.}
\item The asymptotic behavior of the detrended variables is
\begin{equation*}
    \lim_{t\to\infty} G^{-t}(c_t^y,c_t^o,P_t)=\begin{cases*}
        \left(\frac{a+b}{1+\beta},\frac{\beta(a+b)}{1+\beta},\frac{\beta a-b}{1+\beta}\right) & if $\beta a>b$ and $P_0=\frac{\beta a-b}{1+\beta}$,\\
        (a,b,0) & otherwise.
    \end{cases*}
\end{equation*}
\end{enumerate}
\end{prop}

Proposition \ref{prop:example_eq} shows that, when $\beta a>b$, we are in an uncomfortable situation where there is a continuum of equilibria. Furthermore, each equilibrium has a different consumption allocation (the economy features real equilibrium indeterminacy) because $P_t$ is strictly increasing in $P_0$ by \eqref{eq:Pt} and the consumption of the young is $aG^t-P_t$. Thus, it is not obvious which equilibrium we should select.

One possibility for equilibrium selection is to look for equilibria that are more efficient. The following proposition shows that all equilibria can be Pareto ranked.

\begin{prop}\label{prop:Pareto_rank}
Suppose $\beta a>b$ and let $E(P_0)$ be the equilibrium with initial asset price $P_0\in \left[0,\frac{\beta a-b}{1+\beta}\right]$. If $P_0<P_0'$, then $E(P_0')$ strictly Pareto dominates $E(P_0)$. Furthermore, $E\left(\frac{\beta a-b}{1+\beta}\right)$ is Pareto efficient.
\end{prop}

By Proposition \ref{prop:Pareto_rank}, the only equilibrium that is not Pareto dominated (and is itself efficient) is the one corresponding to $P_0=\frac{\beta a-b}{1+\beta}$ and hence
\begin{equation}
    (c_t^y,c_t^o,P_t)=\left(\frac{a+b}{1+\beta}G^t,\frac{\beta(a+b)}{1+\beta}G^t,\frac{\beta a-b}{1+\beta}G^t\right). \label{eq:Samuelson_BG}
\end{equation}
Hence, one may argue that $E\left(\frac{\beta a-b}{1+\beta}\right)$ is the most natural equilibrium because it is the most efficient. However, this equilibrium (and any other with $P_0>0$) exhibits an asset price bubble. Another natural equilibrium may be $E(0)$, in which there is no asset price bubble, but it is the most inefficient equilibrium. In any case, using Pareto efficiency as an equilibrium selection criterion is not convincing because in incomplete-market models like \S\ref{sec:lev}, all equilibria are inefficient.

We now argue that a slight perturbation of the economy leads to a unique equilibrium, and that only one of the equilibria of the original economy can be achieved as the perturbation approaches zero, thereby enabling us to select a unique equilibrium.\footnote{We note that equilibrium uniqueness is specific to this example economy with log utility. In OLG models, it is well known that cyclic or chaotic equilibria may arise if relative risk aversion is sufficiently high \citep{Grandmont1985} so that the income effect dominates the substitution effect. In general, equilibrium uniqueness requires strong conditions \citep{TodaWalsh2024JME}.} The idea is to ``inject'' a small dividend into the asset. Instead of a pure bubble asset, suppose that the asset pays the dividend $D_t=D_0G_d^t$, where the initial dividend $D_0>0$ is sufficiently small and the dividend growth rate satisfies $G_d<G$ so that the dividend is asymptotically negligible relative to endowments. Combining the Euler equation \eqref{eq:euler_young} (including the dividend) and the no-arbitrage condition
\begin{equation}
    R_t=\frac{P_{t+1}+D_{t+1}}{P_t}, \label{eq:noarbitrage}
\end{equation}
the equilibrium condition becomes
\begin{align}
    &1=\beta \frac{P_{t+1}+D_{t+1}}{P_t}\frac{aG^t-P_t}{bG^{t+1}+P_{t+1}+D_{t+1}}\notag \\
    \iff &P_{t+1}=\frac{bG^{t+1}P_t}{\beta aG^t-(1+\beta)P_t}-D_{t+1}. \label{eq:PtD_diff}
\end{align}
Although it is not generally possible to solve the difference equation \eqref{eq:PtD_diff} explicitly since it is nonlinear and non-autonomous (the equation explicitly depends on time), we may study its qualitative property by applying results from the theory of dynamical systems. To this end, introduce the auxiliary variables $\xi_{1t}\coloneqq G^{-t}P_t$ and $\xi_{2t}=(G_d/G)^t D_0$. Then the one-dimensional non-autonomous nonlinear difference equation \eqref{eq:PtD_diff} can be converted to the two-dimensional autonomous nonlinear difference equation
\begin{subequations}\label{eq:xt_autonomous}
\begin{align}
    \xi_{1,t+1}&=\frac{b\xi_{1t}}{\beta a-(1+\beta)\xi_{1t}}-\frac{G_d}{G}\xi_{2t}, \label{eq:xt_autonomous1}\\
    \xi_{2,t+1}&=\frac{G_d}{G}\xi_{2t}. \label{eq:xt_autonomous2}
\end{align}
\end{subequations}
Let us write \eqref{eq:xt_autonomous} as $\xi_{t+1}=\phi(\xi_t)$, where $\xi_t=(\xi_{1t},\xi_{2t})$ and
\begin{equation}
    \phi(\xi_1,\xi_2)=\begin{bmatrix}
        \frac{b\xi_1}{\beta a-(1+\beta)\xi_1}-\frac{G_d}{G}\xi_2\\
        \frac{G_d}{G}\xi_2
    \end{bmatrix}. \label{eq:h}
\end{equation}
Solving $\xi=\phi(\xi)$, since $G_d<G$ we find that $\phi$ has two fixed points
\begin{equation}
    \xi_f^*\coloneqq (0,0)\quad \text{and} \quad \xi_b^*\coloneqq \left(\frac{\beta a-b}{1+\beta},0\right), \label{eq:steady_states}
\end{equation}
where the subscripts $f$ and $b$ refer to ``fundamental'' and ``bubbly''.

The following proposition shows that, by appropriately choosing the dividend growth rate $G_d$, we can eliminate the fundamental steady state. Furthermore, the bubbly steady state is locally determinate in the sense that there exists a unique equilibrium path converging to the bubbly steady state.

\begin{prop}\label{prop:determinacy_Samuelson}
Suppose $\beta a>b$ and $G_d\in \left(\frac{bG}{\beta a},G\right)$. Then the following statements are true.
\begin{enumerate}
    \item\label{item:determinacy_Samuelson1} There exist no equilibrium paths converging to the fundamental steady state $\xi_f^*$.
    \item\label{item:determinacy_Samuelson2} For any initial dividend $D_0>0$, there exists a unique equilibrium path $\set{\xi_t}_{t=0}^\infty$ converging to the bubbly steady state $\xi_b^*$.
\end{enumerate}
\end{prop}

The proof is technical and is deferred to the appendix. However, it is very simple to see that there cannot be any equilibrium paths converging to the fundamental steady state $\xi_f^*$. To see this, suppose an equilibrium with $\xi_t\to \xi_f^*$ exists and let $R_t>0$ be the gross risk-free rate implied by \eqref{eq:euler_young}. Then
\begin{align*}
    R_t&=\frac{1}{\beta}\frac{bG^{t+1}+P_{t+1}+D_{t+1}}{aG^t-P_t}\\
    &=\frac{G}{\beta}\frac{b+\xi_{1,t+1}+\xi_{2,t+1}}{a-\xi_{1t}}\to \frac{bG}{\beta a}<G_d
\end{align*}
as $t\to\infty$. Since the risk-free rate is asymptotically lower than the dividend growth rate, the present value of the dividend stream is infinite, which is impossible in equilibrium.

\section{Equilibrium selection by dividend injection}\label{sec:inject}

The idea of Proposition \ref{prop:determinacy_Samuelson} to eliminate a fundamental equilibrium (an equilibrium in which the bubble asset has no value) by adding asymptotically negligible but positive dividends to the bubble asset is fairly general. In this section, we discuss this ``dividend injection'' procedure in a general setting and present applications to canonical models in the literature.

\subsection{Existence of bubbly equilibrium}\label{subsec:exist}

We start with a system of reduced-form equations
\begin{subequations}\label{eq:system}
\begin{align}
    W'&=G(R)W, \label{eq:system_dynamics}\\
    s(R)W&=B. \label{eq:system_clear}
\end{align}
\end{subequations}
Here $W$ ($W'$) is the current (next period's) ``aggregate wealth'';\footnote{Here and elsewhere, we use quotation marks because the interpretation is irrelevant, as long as the model satisfies Assumptions \ref{asmp:G} and \ref{asmp:s} below.} $R$ is the gross risk-free rate; $G(R)$ is the gross growth rate of the economy given the interest rate; $s(R)$ is the ``saving rate'' or the demand for risk-free assets per unit of wealth, given the interest rate; and $B$ is the market capitalization of the risk-free (pure bubble) asset. Thus \eqref{eq:system_dynamics} and \eqref{eq:system_clear} can be interpreted as the aggregate resource constraint and market clearing condition, respectively.

Although the system of equations \eqref{eq:system} is reduced-form, it can often be derived from first principles. For instance, suppose agents have homothetic preferences and solve an optimal consumption-portfolio problem. In this setting, because consumption and the demand for particular assets are proportional to wealth due to homotheticity, we obtain equations of the form \eqref{eq:system} after aggregating across agents. Although this argument is somewhat informal, we shall present specific examples in the remainder of the paper.

Following the literature on rational bubbles \citep{MartinVentura2018,HiranoToda2024JME}, we focus on (asymptotically) balanced growth equilibria. By \eqref{eq:system_dynamics}, $W$ grows at the gross growth rate $G(R)$. By \eqref{eq:system_clear}, $B$ is proportional to $W$, so it also grows at rate $G(R)$. However, because $B$ is the market capitalization of pure bubble asset, the absence of arbitrage implies that its rate of return (hence growth rate) equals the gross risk-free rate $R$. Therefore, for a balanced growth equilibrium to exist, there are only two possibilities. One is $G(R)=R$ for the consistency of the growth rates. Another is $G(R)\neq R$ but $B=0$ so that the heterogeneity in returns is consistent with balanced growth, which we refer to as a fundamental or bubbleless equilibrium.


In either case, let $G=G(R)$ be the equilibrium growth rate of the economy. In equilibrium, it could be $R\ge G$ or $R<G$, where $R\neq G$ forces $B=0$. We now argue that an equilibrium with $R<G$ is implausible. To see this, 
instead of the pure bubble asset, let the asset pay a dividend $D_t=D_0 G_d^t$ at time $t$, where the dividend growth rate satisfies the ``bubble necessity condition'' \citep{HiranoToda2025JPE}
\begin{equation}
	R<G_d<G. \label{eq:necessity}
\end{equation}
Then the dividend-to-wealth ratio is
\begin{equation*}
    \frac{D_t}{W_t}=\frac{D_0}{W_0}(G_d/G)^t\to 0
\end{equation*}
as $t\to\infty$ because $G_d<G$ by \eqref{eq:necessity}. Hence, the asset is asymptotically irrelevant, and we may conjecture that the perturbed economy will converge to the original equilibrium. However, if such an equilibrium exists, then the asset price must be infinite because the present value of time $t$ dividend, which is approximately $D_0(G_d/R)^t$, diverges to infinity because $G_d>R$, which is impossible.\footnote{\citet{HiranoToda2025JPE} make this argument precise.} This argument shows that an equilibrium with $R<G$ is fragile in the sense that if we introduce asymptotically negligible dividends (with growth rate between the interest rate and the economic growth rate), the original equilibrium can no longer be supported as a steady state (asymptotically balanced growth path) in the perturbed economy.

But if an equilibrium with $R<G$ is implausible, are there other (plausible) equilibria? The answer is typically yes. Suppose that the growth rate of the economy satisfies $G(R)<R$ for large enough $R$. This is a natural condition: it holds trivially in models with exogenous growth (because $G(R)=G$ is constant); in endogenous growth models, high interest rates typically reduce investment and, hence, growth. Thus if there exists a fundamental equilibrium with $R_f<G(R_f)$ and the growth rate satisfies $R>G(R)$ for large enough $R$, (assuming continuity) there exists $R_b\in (R_f,\infty)$ such that $R_b=G(R_b)$. To construct such an equilibrium, introduce a unit supply of pure bubble asset with price $P_t$, which can be thought of as a limit of a vanishingly small dividend-paying asset. Then \eqref{eq:system} becomes
\begin{subequations}\label{eq:bubble}
\begin{align}
    W_{t+1}&=R_bW_t, \label{eq:bubble_dynamics}\\
    s(R_b)W_t&=P_t. \label{eq:bubble_clear}
\end{align}
\end{subequations}
Using \eqref{eq:bubble}, the gross return on the bubble asset is
\begin{equation*}
    \frac{P_{t+1}}{P_t}=\frac{s(R_b)W_{t+1}}{s(R_b)W_t}=\frac{W_{t+1}}{W_t}=R_b,
\end{equation*}
so the no-arbitrage condition holds and we indeed have an equilibrium.

We now formalize the preceding argument.

\begin{asmp}\label{asmp:G}
Given the gross risk-free rate $R>0$, the aggregate dynamics satisfies $W_{t+1}=G(R)W_t$, where $G:(0,\infty)\to (0,\infty)$ is continuous and satisfies $G(R)<R$ for large enough $R$.
\end{asmp}

\begin{asmp}\label{asmp:s}
Given the gross risk-free rate $R>0$ and aggregate state variable $W_t$, the aggregate demand for the risk-free asset is $s(R)W_t$, where $s:(0,\infty)\to \R$ is continuous.
\end{asmp}

Under the maintained assumptions, the following theorem shows that whenever the interest rate is lower than the economic growth rate in the fundamental equilibrium, we may construct a bubbly equilibrium in which the interest rate equals the economic growth rate.

\begin{thm}[Existence of bubbly equilibrium]\label{thm:exist}
Consider an economy satisfying Assumptions \ref{asmp:G} and \ref{asmp:s} and suppose there exists an interest rate $R_f$ with $R_f<G(R_f)$ and $s(R_f)=0$, which constitutes a fundamental equilibrium. Then there exists $R_b\in (R_f,\infty)$ with $R_b=G(R_b)$. If $s(R_b)>0$, then there exists a bubbly equilibrium with interest rate $R_b$ and asset price $P_t=s(R_b)W_t$.
\end{thm}

\begin{proof}
Immediate from the argument in the main text.
\end{proof}

Whenever the assumptions of Theorem \ref{thm:exist} are satisfied, we may eliminate the fundamental steady state by introducing an asset with small dividends with an appropriate growth rate, just as we did in Proposition \ref{prop:determinacy_Samuelson}. This procedure will often lead to a unique equilibrium path converging to the bubbly steady state. However, the exact implementation depends on the model specifics. In the next few subsections, we apply Theorem \ref{thm:exist} and the elimination of fundamental equilibria via dividend injection to canonical models of rational bubbles with overlapping generations or infinite horizon, with or without production.

\subsection{OLG endowment economy \texorpdfstring{\citep{Samuelson1958}}{}}\label{subsec:Samuelson}

We revisit the OLG model in \S\ref{sec:example}. To apply Theorem \ref{thm:exist}, define the ``aggregate wealth'' $W_t$ by the income of the young, so $W_t=aG^t$. Then $W_{t+1}=GW_t$, so $G(R)=G$ in \eqref{eq:system_dynamics} is constant. We next solve for the ``saving rate'' $s(R)$ in \eqref{eq:system_clear}. Letting $R$ be the gross risk-free rate and $s$ be the savings of the young, the Euler equation implies
\begin{equation*}
    1=\beta R\left(\frac{bG^{t+1}+Rs}{aG^t-s}\right)^{-1}\iff s=\frac{G^t}{1+\beta}\left(\beta a-\frac{bG}{R}\right).
\end{equation*}
Therefore, we may define the ``saving rate'' in \eqref{eq:system_clear} by
\begin{equation*}
    s(R)=\frac{s}{aG^t}=\frac{1}{1+\beta}\left(\beta-\frac{bG}{aR}\right).
\end{equation*}
Thus the fundamental equilibrium condition is
\begin{equation*}
    s(R)=0\iff R_f=\frac{bG}{\beta a},
\end{equation*}
which is precisely \eqref{eq:R_autarky}. The low interest condition is therefore
\begin{equation*}
    G>R_f=\frac{bG}{\beta a}\iff \beta a>b,
\end{equation*}
which is precisely the existence condition for a bubbly equilibrium in Proposition \ref{prop:example_eq}. Clearly $s(R)>0$ for $R>R_f$, so all assumptions of Theorem \ref{thm:exist} are satisfied. The bubbly equilibrium condition is $R_b=G$ and the asset price is
\begin{equation*}
    P_t=s(R_b)W_t=\frac{\beta a-b}{1+\beta}G^t,
\end{equation*}
which corresponds to the balanced growth bubbly equilibrium \eqref{eq:Samuelson_BG}.

\subsection{OLG production economy \texorpdfstring{\citep{Tirole1985}}{}}\label{subsec:Tirole}

\citet{Tirole1985} introduces a pure bubble asset in the \citet{Diamond1965} overlapping generations neoclassical growth model. Since the model is well known, we only describe it briefly.\footnote{See \citet[\S 5.2]{BlanchardFischer1989} for a textbook treatment.} The number of young agents at time $t$ is $N_t=N_0G^t$. The young maximize utility $U(c_t^y,c_{t+1}^o)$ subject to the budget constraint $c_t^y+c_{t+1}^o/R_t=\omega_t$, where $R_t$ is the gross risk-free rate from time $t$ to $t+1$ and $\omega_t$ is the wage. The firm has a constant-returns-to-scale production function $F(K,N)$, hires labor $N$ to maximize the profit $F(K,N)-\omega N$, and distributes the profit to capital owners. Letting $K_t$ be aggregate capital, $k_t=K_t/N_t$ be capital per capita (capital-labor ratio), and $f(k)\coloneqq F(k,1)$, profit maximization implies
\begin{equation*}
    \omega_t=\omega(k_t)\coloneqq f(k_t)-k_tf'(k_t).
\end{equation*}
The gross risk-free rate between time $t-1$ and $t$ is then
\begin{equation}
    R_{t-1}=\frac{F(K_t,N_t)-\omega_tN_t}{K_t}=f'(k_t). \label{eq:R_Tirole}
\end{equation}
Since $\omega_t$ and $R_{t-1}$ are functions of $k_t$, the consumption of the young can be written as $c_t^y=c^y(k_t,k_{t+1})$. Hence the aggregate demand for the risk-free asset (in excess of capital) at time $t$ is
\begin{equation}
    S_t\coloneqq N_t(\omega(k_t)-c^y(k_t,k_{t+1}))-K_{t+1}=N_t(\omega(k_t)-c^y(k_t,k_{t+1})-Gk_{t+1}). \label{eq:rfdemand_Tirole}
\end{equation}

Let us show that this model fits in the framework described in Theorem \ref{thm:exist}. Suppose the capital-labor ratio $k$, wage $\omega$, and gross risk-free rate $R$ are constant over time. As in \S\ref{subsec:Samuelson}, define the ``aggregate wealth'' by the aggregate income of the young, so $W_t=\omega N_0G^t$. Then $W_{t+1}=GW_t$, so Assumption \ref{asmp:G} holds. Using \eqref{eq:rfdemand_Tirole}, the ``saving rate'' is
\begin{equation}
    s(R)\coloneqq \frac{S_t}{W_t}=1-\frac{c^y(k,k)+Gk}{\omega(k)}, \label{eq:DR_Tirole}
\end{equation}
where $k=(f')^{-1}(R)$. Therefore, Assumption \ref{asmp:s} holds. We obtain the fundamental equilibrium interest rate by setting $s(R)=0$. Letting $k_f$ be the corresponding capital per capita, we have $R_f=f'(k_f)$. If the low interest condition $R_f<G$ holds, under standard monotonicity, concavity, and Inada conditions on $f$, there exists a unique $k_b<k_f$ such that $R_f=f'(k_f)<f'(k_b)=G\eqqcolon R_b$. If $s(G)>0$, there exists a bubbly equilibrium with asset price
\begin{equation*}
    P_t=s(R_b)W_t=(\omega(k_b)-c^y(k_b,k_b)-Gk_b)N_0G^t.
\end{equation*}

We now show that dividend injection eliminates the fundamental equilibrium. As in \S\ref{sec:example}, introduce an asset in unit supply with dividend $D_t=D_0G_d^t$, where the bubble necessity condition $R_f<G_d<G$ \eqref{eq:necessity} holds. Combining the no-arbitrage condition \eqref{eq:noarbitrage}, the definition of the risk-free rate \eqref{eq:R_Tirole}, and the per capita demand for the asset \eqref{eq:DR_Tirole}, the equilibrium is characterized by the system of equations
\begin{subequations}\label{eq:system_Tirole}
    \begin{align}
        f'(k_{t+1})&=\frac{P_{t+1}+D_{t+1}}{P_t}, \label{eq:system_Tirole1}\\
        \omega(k_t)-c^y(k_t,k_{t+1})-Gk_{t+1}&=P_t/N_t, \label{eq:system_Tirole2}\\
        D_{t+1}&=G_dD_t. \label{eq:system_Tirole3}
    \end{align}
\end{subequations}
Here \eqref{eq:system_Tirole1} is the no-arbitrage condition between capital and the asset, \eqref{eq:system_Tirole2} is the condition that the asset absorbs the aggregate savings, and \eqref{eq:system_Tirole3} governs the evolution of dividends. To study asymptotically balanced growth paths, define the auxiliary variables $\xi_t=(\xi_{1t},\xi_{2t},\xi_{3t})=(k_t,P_t/N_t,D_t/N_t)$. Then the system of equations \eqref{eq:system_Tirole} can be written as $\Phi(\xi_t,\xi_{t+1})=0$, where
\begin{subequations}\label{eq:Phi_Tirole}
\begin{align}
    \Phi_1(\xi,\eta)&=\eta_2+\eta_3-\frac{\xi_2}{G}f'(\eta_1),\\
    \Phi_2(\xi,\eta)&=\xi_2+G\eta_1+c^y(\xi_1,\eta_1)-\omega(\xi_1),\\
    \Phi_3(\xi,\eta)&=\eta_3-\frac{G_d}{G}\xi_3.
\end{align}
\end{subequations}
The fundamental and bubbly steady states are given by
\begin{equation*}
    \xi_f^*\coloneqq (k_f,0,0)\quad \text{and} \quad \xi_b^*\coloneqq \left(k_b,\omega(k_b)-c^y(k_b,k_b)-Gk_b,0\right).
\end{equation*}
We now present a result establishing the nonexistence of fundamental equilibria after dividend injection and a sufficient condition for the existence of a locally determinate bubbly equilibrium. Before doing so, recall that with the utility function $U(c_1,c_2)$, the elasticity of (intertemporal) substitution $\varepsilon$  is defined by
\begin{equation}
    \frac{1}{\varepsilon}=-\frac{\partial \log(U_1/U_2)}{\partial \log (c_1/c_2)}, \label{eq:EIS}
\end{equation}
see for instance the discussion in \citet{FlynnSchmidtToda2023TE}.

\begin{thm}\label{thm:determinacy_Tirole}
Suppose $R_f=f'(k_f)<G$, $s(G)>0$, and $G_d\in (R_f,G)$. Then the following statements are true.
\begin{enumerate}
\item\label{item:determinacy_Tirole1} There exist no equilibrium paths converging to the fundamental steady state $\xi_f^*$.
\item\label{item:determinacy_Tirole2} If the elasticity of intertemporal substitution satisfies
\begin{equation}
    \varepsilon>1+\frac{G^2}{f''}\frac{\omega}{c^y(\omega-c^y)},\label{eq:EIS_cond}
\end{equation}
where all functions are evaluated at $k_b$, then the bubbly steady state $\xi_b^*$ is locally determinate. More precisely, for any $(k_0,D_0)$ sufficiently close to $(k_b,0)$, there exists a unique equilibrium path $\set{\xi_t}_{t=0}^\infty$ converging to $\xi_b^*$.
\end{enumerate}
\end{thm}

In the textbook treatment of the \citet{Tirole1985} model, \citet[p.~232]{BlanchardFischer1989} state ``The knife-edge case where the bubble is just such that the economy is on the saddle point path, though unlikely, is interesting''. Theorem \ref{thm:determinacy_Tirole}\ref{item:determinacy_Tirole1} reveals that the bubbly saddle point path is actually not a knife-edge case at all. It is, in fact, the stable ``bubbleless'' or fundamental equilibrium that is a knife-edge case because any small perturbations to the dividends will eliminate this equilibrium as well as any equilibria converging to the fundamental steady state.

To our knowledge, most authors discuss the saddle path property of the bubbly steady state without showing it from first principles.\footnote{For instance, \citet[p.~268, Endnote 16]{BlanchardFischer1989} state ``Care must be taken in using a phase diagram to analyze the dynamics of a \emph{difference} equation system. [\ldots] Thus we must check in this case whether the system is indeed saddle point stable [\ldots]. This check is left to the reader''.} As we can see from Theorem \ref{thm:determinacy_Tirole}\ref{item:determinacy_Tirole2}, the saddle path property (local determinacy) of the bubbly steady state does not come for free.\footnote{See \citet{Calvo1978} and \citet{Woodford1984} for a discussion of equilibrium indeterminacy in OLG models.} Note that the right-hand side of \eqref{eq:EIS_cond} is strictly less than 1 because $f''<0$. Thus, the sufficient condition for local determinacy is that the elasticity of substitution is not too much below 1, which resembles the characterization of equilibrium uniqueness in general equilibrium models \citep{TodaWalsh2017ETB,TodaWalsh2024JME}.

As a concrete example, suppose that the production function is Cobb-Douglas with a constant depreciation rate, so $f(k)=Ak^\alpha+(1-\delta)k$ for some $A>0$, $\alpha\in (0,1)$, and $\delta\in [0,1]$. Then at the bubbly steady state $k=k_b$, we have
\begin{align*}
    G=f'(k)&=A\alpha k^{\alpha-1}+1-\delta,\\
    f''(k)&=A\alpha(\alpha-1)k^{\alpha-2},\\
    \omega(k)&=f(k)-kf'(k)=A(1-\alpha)k^\alpha.
\end{align*}
Using $c^y(\omega-c^y)\le \omega^2/4$, the lower bound in \eqref{eq:EIS_cond} can be bounded above as
\begin{equation*}
    1+\frac{G^2}{f''}\frac{\omega}{c^y(\omega-c^y)}\le 1+\frac{4G^2}{\omega f''}=1-\frac{4\alpha}{(1-\alpha)^2}\left(1-\frac{1-\delta}{G}\right)^{-2}.
\end{equation*}
Therefore, we obtain a sufficient condition for local determinacy that involves only exogenous parameters $G,\alpha,\delta$. See \citet{PhamTodaTirole,PhamToda2026ECMA} for a more in-depth discussion of \citet{Tirole1985}'s model.

\subsection{Infinite-horizon endowment economy \texorpdfstring{\citep{Kocherlakota1992}}{}} \label{subsec:Kocherlakota}

\citet[Example 1]{Kocherlakota1992} considers a two-agent economy (with relative risk aversion equal to 2) in which aggregate endowment grows at a constant rate, the individual endowments alternate between high and low levels every period, and agents trade a pure bubble asset subject to a shortsales constraint. This example builds on \citet{Bewley1980}, which features no growth but a general utility function. Here we follow the exposition in \citet[\S3.2]{HiranoToda2024JME}.

The two agents have constant relative risk aversion (CRRA) utility
\begin{equation*}
    \sum_{t=0}^\infty \beta^t\frac{c_t^{1-\gamma}}{1-\gamma}.
\end{equation*}
The individual endowments are $(e_{1t},e_{2t})=(aG^t,bG^t)$ when $t$ is odd and $(bG^t,aG^t)$ when $t$ is even, where $a>b>0$ and $G>0$ is the growth rate of the aggregate endowment. Call the agent with endowment $aG^t$ ($bG^t$) ``rich'' (``poor''). The only asset in the economy is a pure bubble asset in unit supply, which is initially held by the poor agent.

As in \S\ref{subsec:Samuelson}, define the ``aggregate wealth'' by the income of the rich agent, so $W_t=aG^t$ and hence $W_{t+1}=GW_t$. Letting $R$ be the gross risk-free rate and $s$ be the savings of the rich agent (who wishes to save and hence the shortsales constraint does not bind), the Euler equation implies
\begin{equation*}
    1=\beta R\left(\frac{bG^{t+1}+Rs}{aG^t-s}\right)^{-\gamma}\iff s=\frac{(\beta R)^{1/\gamma}-\frac{b}{a}G}{R+(\beta R)^{1/\gamma}}aG^t\eqqcolon s(R)W_t.
\end{equation*}
We obtain the fundamental equilibrium interest rate by setting $s(R)=0$. The low interest condition is therefore
\begin{equation}
    G>R_f=\frac{1}{\beta}\left(\frac{b}{a}G\right)^\gamma\iff b<(\beta G^{1-\gamma})^{1/\gamma}a, \label{eq:low_interest_Kocherlakota}
\end{equation}
which is precisely the bubble existence condition in \citet*[Proposition 2]{HiranoToda2024JME}.\footnote{In addition to the low interest condition \eqref{eq:low_interest_Kocherlakota}, it is necessary and sufficient for equilibrium existence that $\beta G^{1-\gamma}<1$, which implies the Euler inequality for the poor agent and the transversality condition for individual optimality. See the discussion in \citet*[\S3.2]{HiranoToda2024JME} for details.} Under this condition, all assumptions of Theorem \ref{thm:exist} are satisfied. The bubbly equilibrium condition is $R_b=G$ and the asset price is
\begin{equation*}
    P_t=s(R_b)W_t=\frac{(\beta G^{1-\gamma})^{1/\gamma}a-b}{1+(\beta G^{1-\gamma})^{1/\gamma}}G^t.
\end{equation*}

\section{Entrepreneurial economy with leverage}\label{sec:lev}

As a more substantive application, we consider an infinite-horizon production economy in which investment opportunities arrive stochastically and agents (entrepreneurs) can borrow subject to the leverage constraint. The model we present here is an extension of \citet{Kocherlakota2009} with labor productivity growth and leverage, which is also related to the model of \citet*{HiranoJinnaiTodaLeverage} and \citet[\S V.B]{HiranoToda2025JPE}.\footnote{In our analysis here, the growth rate of the economy is (asymptotically) exogenously determined by labor productivity growth, while in \citet*{HiranoJinnaiTodaLeverage} and \citet[\S V.B]{HiranoToda2025JPE}, it is endogenously determined depending on the financial leverage or aggregate productivity.}


\subsection{Setup}\label{subsec:lev_setup}

We consider an infinite-horizon economy with a homogeneous good. Time is discrete and denoted by $t=0,1,\dotsc$.

\paragraph{Agents}
The economy is populated by two types of agents, entrepreneurs and workers. Workers supply labor inelastically and consume the entire wage (hand-to-mouth) for simplicity. The aggregate labor supply at time $t$ is $G^t$, where $G>1$ can be interpreted as the growth rate of the population or labor productivity. There is a unit mass of entrepreneurs (agents) indexed by $i\in I\coloneqq [0,1]$. A typical agent has utility function
\begin{equation}
    \E_0\sum_{t=0}^\infty \beta^t\log c_t, \label{eq:log_utility}
\end{equation}
where $\beta\in (0,1)$ is the discount factor and $c_t\ge 0$ is consumption.

\paragraph{Investment opportunities and production}

Each period, an investment opportunity arrives to an agent with probability $\pi\in (0,1)$, which is independent across agents and time. If an agent has an investment opportunity at time $t-1$ (is a high-productivity ($H$) type) and decides to install capital $k$, then the agent can produce goods using the production function $F(k,n)$ at time $t$, where $n$ is labor input. As usual, suppose that $F$ is homogeneous of degree 1, strictly quasi-concave, continuously differentiable, and satisfies the Inada conditions. If an agent does not have an investment opportunity (is a low-productivity ($L$) type), there is no production. Therefore, the production function at time $t$ can be written as $F_t(k,n)=A_{t-1}F(k,n)$, where $A_{t-1}=1$ if the agent has an investment opportunity at time $t-1$ and $A_{t-1}=0$ otherwise. Regardless of the investment opportunity, capital depreciates at rate $\delta\in [0,1]$.

\paragraph{Assets and budget constraint}

There is a risk-free asset in zero net supply and a pure bubble asset (land) in unit supply. Let $R_t$ be the gross risk-free rate between time $t$ and $t+1$, and let $P_t$ be the land price. Suppressing the individual subscript, the budget constraint of a typical agent is
\begin{equation}
    c_t+k_{t+1}+P_tx_t+b_t=A_{t-1}F(k_t,n_t)+(1-\delta)k_t-\omega_tn_t+P_tx_{t-1}+R_{t-1}b_{t-1},\label{eq:budget}
\end{equation}
where $k_t,n_t\ge 0$ are capital and labor inputs at time $t$, $A_{t-1}$ is an indicator for investment opportunity, $\omega_t\ge 0$ is the wage, and $x_t,b_t$ are land and bond holdings.

\paragraph{Leverage constraint}

Agents are subject to the leverage constraint
\begin{equation}
    k_{t+1}\le \lambda(k_{t+1}+P_tx_t+b_t),\label{eq:leverage}
\end{equation}
where $\lambda\ge 1$ is the leverage limit. Here $k_{t+1}+P_tx_t+b_t$ is total financial asset (``equity'') of the agent. The leverage constraint \eqref{eq:leverage} implies that total investment in the production technology cannot exceed some multiple of total
equity, which is essentially identical to the constraint in \citet*{HiranoJinnaiTodaLeverage}. If $\lambda=1$ (no leverage allowed), then \eqref{eq:leverage} reduces to
\begin{equation}
    k_{t+1}\le k_{t+1}+P_tx_t+b_t \iff -b_t\le P_tx_t,\label{eq:kocher_borrow}
\end{equation}
so borrowing ($-b_t$) must be collateralized by land, which is exactly the constraint considered in \citet{Kocherlakota2009}.

\paragraph{Equilibrium}

The economy starts at $t=0$ with some initial distribution of capital and land $\set{(k_{i0},x_{i,-1})}_{i\in I}$, where $(k_{i0},x_{i,-1})>0$ for all $i$. The definition of a rational expectations equilibrium is standard.

\begin{defn}[Rational expectations equilibrium]
Given the initial condition $\set{(k_{i0},x_{i,-1})}_{i\in I}$, a \emph{rational expectations equilibrium} consists of wages $\set{\omega_t}_{t=0}^\infty$, land prices $\set{P_t}_{t=0}^\infty$, interest rates $\set{R_t}_{t=0}^\infty$, and allocations $\set{(c_{it},k_{it},n_{it},x_{it},b_{it})_{i\in I}}_{t=0}^\infty$ such that the following conditions hold.
\begin{enumerate}
\item (Individual optimization) Agents maximize the utility \eqref{eq:log_utility} subject to the budget constraint \eqref{eq:budget} and the leverage constraint \eqref{eq:leverage}.
\item (Labor market clearing) For all $t$, we have $\int_I n_{it}\diff i=G^t$.
\item (Land market clearing) For all $t$, we have $\int_I x_{it}\diff i=1$.
\item (Bond market clearing) For all $t$, we have $\int_I b_{it}\diff i=0$.
\end{enumerate}
\end{defn}

\subsection{Equilibrium analysis}\label{subsec:lev_eq}

We analyze the equilibrium dynamics.

\paragraph{Labor demand and profit}
For an $L$ agent, it is clearly optimal to choose $n_t=0$. Consider an $H$ agent. Letting $f(k)=F(k,1)$, we have $F(k,n)=nf(k/n)$ by homogeneity. Therefore, the first-order condition for profit maximization is
\begin{equation}
    \omega_t=f(y_t)-y_tf'(y_t),\label{eq:labor_foc}
\end{equation}
where $y_t=k_t/n_t$ is the capital-labor ratio. Since $(f(y)-yf'(y))'=-yf''(y)>0$, there exists a unique $y_t=y(\omega_t)>0$ satisfying \eqref{eq:labor_foc}. The labor demand is then
\begin{equation}
    n_t=k_t/y_t. \label{eq:nt}
\end{equation}
The maximized profit is
\begin{equation*}
    F(k_t,n_t)-\omega n_t=\frac{k_t}{y}[f(y)-(f(y)-yf'(y))]=f'(y)k_t.
\end{equation*}

\paragraph{Consumption and investment}
Let $w_t$ be the beginning-of-period wealth defined by the right-hand side of \eqref{eq:budget} maximized over $n_t$, so
\begin{equation}
    w_t\coloneqq (A_{t-1}f'(y_t)+1-\delta)k_t+P_tx_{t-1}+R_{t-1}b_{t-1}. \label{eq:wt}
\end{equation}
Due to log utility \eqref{eq:log_utility}, the optimal consumption rule is $c_t=(1-\beta)w_t$. Therefore, the agent saves $\beta w_t$, which must equal $k_{t+1}+P_tx_t+b_t$ by accounting.

How agents allocate savings to capital ($k_{t+1}$), land ($x_t$), and bonds ($b_t$) depends on the interest rate and investment opportunities. In equilibrium, in order for investment to occur (so that there is demand for labor), the return on capital for $H$ agents must exceed the risk-free rate. Furthermore, since $L$ agents can always save at rate $1-\delta$ by holding idle capital, the risk-free rate cannot fall below this value. Therefore, we have the bound
\begin{equation}
    f'(y_{t+1})+1-\delta\ge R_t\ge 1-\delta. \label{eq:R_bound}
\end{equation}
If the return on capital exceeds the risk-free rate, $H$ agents will obviously choose maximal leverage. Therefore, using the leverage constraint \eqref{eq:leverage}, we obtain the optimal investment rule
\begin{subequations}\label{eq:k}
    \begin{align}
        &\text{$H$ agents:} &k_{t+1}&\begin{cases*}
            =\lambda\beta w_t & if $f'(y_{t+1})+1-\delta>R_t$,\\
            \le \lambda\beta w_t & if $f'(y_{t+1})+1-\delta=R_t$,
        \end{cases*} \label{eq:kH}\\
        &\text{$L$ agents:} & k_{t+1}&=\begin{cases*}
            0 & if $R_t>1-\delta$,\\
            \text{arbitrary} & if $R_t=1-\delta$.
        \end{cases*} \label{eq:kL}
    \end{align}
\end{subequations}

\paragraph{Aggregation}

We now aggregate the individual decision rules. Let $W_t=\int_Iw_{it}\diff i$ be the aggregate wealth. Let $K_t^H,K_t^L$ be the aggregate capital holdings of $H,L$ agents, $K_t=K_t^H+K_t^L$ the aggregate capital, and $Y_t=f'(y_t)K_t^H$ the aggregate output. Aggregating individual wealth \eqref{eq:wt} across all agents and using the land and bond market clearing conditions, we obtain
\begin{equation}
    W_t=Y_t+(1-\delta)K_t+P_t.\label{eq:Wt}
\end{equation}
Aggregating individual savings $k_{t+1}+P_tx_t+b_t=\beta w_t$ across all agents, we obtain
\begin{equation}
    K_{t+1}^H+K_{t+1}^L+P_t=\beta W_t. \label{eq:agg_saving}
\end{equation}
Since the fraction of $H$ agents is $\pi$, aggregating the capital investment \eqref{eq:kH} across $H$ agents, we obtain
\begin{equation}
    K_{t+1}^H\le \pi\lambda\beta W_t\eqqcolon \psi\beta W_t,\label{eq:KH}
\end{equation}
with equality if the return on capital strictly exceeds the interest rate. The parameter $\psi\coloneqq \pi\lambda$ can be interpreted as a measure of financial market imperfection. If $\psi\ge 1$ (so $\lambda\ge 1/\pi$), then the $H$ type can absorb all savings. Then the economy reduces to the frictionless case with $K_{t+1}^H=\beta W_t$ and $K_{t+1}^L=P_t=0$, which is uninteresting. 

Therefore, assume $\psi<1$. In this case the $H$ type cannot absorb all savings, so the return on capital must exceed the interest rate and \eqref{eq:KH} holds with equality. Combining \eqref{eq:Wt}, \eqref{eq:agg_saving}, \eqref{eq:KH}, we can solve for the aggregate capital allocation $(K_{t+1}^H,K_{t+1}^L,K_{t+1})$ as a function of predetermined resources $Y_t+(1-\delta)K_t$ and the asset price $P_t$:
\begin{subequations}\label{eq:aggK}
    \begin{align}
        K_{t+1}^H&=\psi\beta(Y_t+(1-\delta)K_t)+\psi\beta P_t, \label{eq:KHP}\\
        K_{t+1}^L&=(1-\psi)\beta(Y_t+(1-\delta)K_t)-(1-\beta+\psi\beta)P_t, \label{eq:KLP} \\
        K_{t+1}&=\beta(Y_t+(1-\delta)K_t)-(1-\beta)P_t. \label{eq:KP}
    \end{align}
\end{subequations}

We can make two observations from \eqref{eq:aggK}. First, as the leverage limit $\lambda$ increases, so does $\psi=\pi\lambda$. Then $K^H$ increases and $K^L$ decreases, so higher leverage allows a more efficient allocation of capital. However, leverage has no direct effect on aggregate capital $K$ because \eqref{eq:KP} does not depend on $\psi$: leverage affects capital only through the asset price $P$. Second, as the asset price $P$ increases, $K^H$ increases and $K^L$ decreases, so an asset price bubble crowds out idle capital and crowds in investment.

\paragraph{Equilibrium}

We focus on balanced growth path equilibria in which the capital-labor ratio $y$, wage $\omega$, and risk-free rate $R$ are constant. Aggregating labor demand \eqref{eq:nt} across $H$ agents and using \eqref{eq:KH}, we obtain the aggregate labor
\begin{equation}
    G^t=N_t=\frac{K_t^H}{y}=\frac{\psi\beta}{y} W_{t-1}. \label{eq:N}
\end{equation}
We next derive the aggregate wealth dynamics similar to \eqref{eq:system_dynamics}. At time $t$, by \eqref{eq:agg_saving} aggregate savings is $\beta W_t$. By \eqref{eq:KH}, fraction $\psi$ of it is invested by $H$ agents as capital and the remaining fraction $1-\psi$ is saved (either in the pure bubble asset or idle capital). Since the return on capital is $f'(y)+1-\delta$ and the risk-free rate is $R$, aggregate wealth evolves according to
\begin{equation}
    W_{t+1}=(\psi(f'(y)+1-\delta)+(1-\psi)R)\beta W_t. \label{eq:W_dynamics}
\end{equation}
By \eqref{eq:N}, aggregate wealth $W_t$ must grow at rate $G$ in equilibrium. Therefore, the coefficient of $W_t$ in \eqref{eq:W_dynamics} must equal $G$ and we obtain the equilibrium condition
\begin{equation}
    G(y,R)\coloneqq \beta(\psi(f'(y)+1-\delta)+(1-\psi)R)=G.\label{eq:GyR}
\end{equation}
In a fundamental equilibrium ($P_t=0$ for all $t$), because $L$ agents hold idle capital, the risk-free rate must be $R=1-\delta$. We thus obtain the following proposition.

\begin{prop}[Fundamental equilibrium]\label{prop:eq_fl}
Suppose $1\le \lambda<1/\pi$ and let $\psi=\pi\lambda$. Then there exists a unique fundamental balanced growth path equilibrium. The equilibrium is characterized by the following equations.
\begin{subequations}
    \begin{align}
        \psi f'(y)+1-\delta&=\frac{1}{\beta}G, \label{eq:eq_yf}\\
        \omega&=f(y)-yf'(y), \label{eq:eq_omegaf}\\
        R&=1-\delta, \label{eq:eq_Rf}\\
        P_t&=0. \label{eq:eq_Pf}
    \end{align}
\end{subequations}
\end{prop}

The proof of Proposition \ref{prop:eq_fl} is immediate by setting $R=1-\delta$ in \eqref{eq:GyR}. The left-hand side of \eqref{eq:eq_yf} is the average return on capital, which pins down the capital-labor ratio $y$. Note that capital generates a return of $1-\delta$ from storage, and only fraction $\psi<1$ is used for production. The wage \eqref{eq:eq_omegaf} is determined by profit maximization. The fundamental interest rate \eqref{eq:eq_Rf} equals the return from storage $1-\delta$.

We next consider a balanced growth path equilibrium with an asset price bubble ($P_t>0$ for all $t$). Using \eqref{eq:Wt}, we may rewrite \eqref{eq:aggK} as
\begin{subequations}
    \begin{align}
        K_{t+1}^H&=\psi\beta W_t, \label{eq:KHW}\\
        K_{t+1}^L&=(1-\psi)\beta W_t-P_t. \label{eq:KLW}
    \end{align}
\end{subequations}
Along a balanced growth path, the labor market clearing condition \eqref{eq:N} forces the aggregate wealth $W_t$ to grow at rate $G$. By the definition of balanced growth, $K_{t+1}^H,K_{t+1}^L,P_t$ must all grow at rate $G$. Then the gross return on land is $P_{t+1}/P_t=G$, so the absence of arbitrage forces the equilibrium interest rate to be $R=G$. Since $G>1>1-\delta$, by \eqref{eq:kL}, we obtain $K_{t+1}^L=0$. That is, $L$ agents buy land instead of holding idle capital with low returns. Therefore, \eqref{eq:KLW} and \eqref{eq:N} imply
\begin{equation}
    P_t=(1-\psi)\beta W_t=\frac{1-\psi}{\psi}yG^{t+1}. \label{eq:PW}
\end{equation}
We thus obtain the following proposition.

\begin{prop}[Bubbly equilibrium]\label{prop:eq_b}
Suppose $1\le \lambda<1/\pi$ and let $\psi=\pi\lambda$. Then there exists a unique bubbly balanced growth path equilibrium. The equilibrium is characterized by the following equations.
\begin{subequations}\label{eq:eq_obj}
    \begin{align}
        f'(y)+1-\delta&=\frac{1-\beta+\psi\beta}{\psi\beta}G,\label{eq:eq_y}\\
        \omega&=f(y)-yf'(y), \label{eq:eq_omega}\\
        R&=G, \label{eq:eq_R}\\
        P_t&=\frac{1-\psi}{\psi}yG^{t+1}. \label{eq:eq_P}
    \end{align}
\end{subequations}
\end{prop}

The proof of Proposition \ref{prop:eq_b} is immediate by setting $R=G$ in \eqref{eq:GyR} and using \eqref{eq:PW}. The left-hand side of \eqref{eq:eq_y} is the equilibrium return on capital, which pins down the capital-labor ratio $y$. The wage \eqref{eq:eq_omega} is determined by profit maximization. The bubbly equilibrium interest rate \eqref{eq:eq_R} equals the economic growth rate, as in Theorem \ref{thm:exist}. Note that in the bubbly equilibrium, the land price \eqref{eq:eq_P} grows at the same rate as the economy.

\subsection{Elimination of fundamental equilibrium}

In our model, we have $R_b=G>1-\delta=R_f$. Since the fundamental equilibrium interest rate $R_f=1-\delta$ is lower than the economic growth rate $G$, we may eliminate the fundamental equilibrium $E_f$ by dividend injection just as we did for the \citet{Samuelson1958} model in Proposition \ref{prop:determinacy_Samuelson}. 

To this end, consider a perturbed economy in which land pays a constant dividend $D>0$ every period. In this case, we need to make two changes to the derivations in \S\ref{subsec:lev_eq}. First, the capital-labor ratio $y$, wage $\omega$, and risk-free rate $R$ are no longer constant over time. Second, we need to include dividends to the aggregate wealth. With these changes, the equilibrium conditions can be summarized as follows:
\begin{subequations}
    \begin{align}
        W_t&=(f'(y_t)+1-\delta)\psi\beta W_{t-1}+P_t+D, \label{eq:eqcond_W}\\
        P_t&=(1-\psi)\beta W_t, \label{eq:eqcond_P}\\
        \frac{\psi\beta}{y_t}W_{t-1}&=G^t. \label{eq:eqcond_L}
    \end{align}
\end{subequations}
Here \eqref{eq:eqcond_W} comes from including dividend to \eqref{eq:Wt}, noting that the return on capital is $f'(y_t)+1-\delta$, and using \eqref{eq:KH}. The condition \eqref{eq:eqcond_P} is identical to \eqref{eq:PW}. The condition \eqref{eq:eqcond_L} is the labor market clearing condition, which comes from \eqref{eq:N}. Combining \eqref{eq:eqcond_W} and \eqref{eq:eqcond_P} to eliminate $P_t$ and then using \eqref{eq:eqcond_L} to eliminate $W_t$ and $W_{t-1}$, the equilibrium dynamics is fully characterized by the nonlinear difference equation involving only the capital-labor ratio
\begin{equation}
    y_{t+1}=\frac{1}{G}\frac{\psi\beta}{1-\beta+\psi\beta}\left(y_t(f'(y_t)+1-\delta)+DG^{-t}\right). \label{eq:eqcond_y}
\end{equation}

The following proposition shows that dividend injection eliminates the (inefficient) fundamental equilibrium. Furthermore, the bubbly equilibrium is locally determinate as long as the elasticity of substitution between capital and labor is not too much below $1/2$. To state this result, recall that the elasticity of substitution $\varepsilon$ is defined by
\begin{equation}
    \frac{1}{\varepsilon}=-\frac{\partial \log (F_K/F_L)}{\partial \log (K/L)}. \label{eq:ES}
\end{equation}

\begin{prop}[Local determinacy of bubbly equilibrium]\label{prop:locdet}
Suppose $1\le \lambda<1/\pi$ and $D>0$. Then the following statements are true.
\begin{enumerate}
    \item There exist no equilibrium paths converging to the fundamental balanced growth path equilibrium $E_f$.
    \item If the bubbly capital-labor ratio $y_b$ satisfies
    \begin{equation}
        y_bf''(y_b)>-2G\frac{1-\beta+\psi\beta}{\psi\beta}, \label{eq:locdet}
    \end{equation}
    for any $(k_0,D)$ sufficiently close to $(y_b,0)$, there exists a unique equilibrium path $\set{y_t}_{t=0}^\infty$ converging to $y_b$. Furthermore, \eqref{eq:locdet} is equivalent to
    \begin{equation}
        \varepsilon>\frac{1}{2}\left(1-\frac{\psi\beta}{1-\beta+\psi\beta}\frac{1-\delta}{G}\right)\frac{1}{1+y_bf'(y_b)/\omega_b}. \label{eq:locdet2}
\end{equation}
\end{enumerate}
\end{prop}

When $y_b$ satisfies \eqref{eq:locdet} (or equivalently \eqref{eq:locdet2}), the bubbly steady state $(y_b,0)$ is stable, and because the initial condition $(k_0,D)$ is exogenous, there exists a unique equilibrium path converging to the steady state. When the reverse inequality to \eqref{eq:locdet} holds, the steady state is unstable and the equilibrium path does not converge unless $k_0$ (or $D$) happens to take a particular value. In this case there is no basis for focusing on the steady state. Proposition \ref{prop:locdet} tells us that focusing on the bubbly steady state is justified if the elasticity of substitution between capital and labor is not too much below $1/2$. In particular, local determinacy holds if the production function is Cobb-Douglas ($\varepsilon=1$).

\section{Further applications}\label{sec:application}

In this section, we provide further applications of equilibrium selection by dividend injection and discuss the related literature.

\subsection{Stochastic bubbles \texorpdfstring{\citep{Weil1987}}{}}\label{subsec:application_stochastic}

So far, we have focused on perfect foresight (deterministic) equilibria, but we can apply similar arguments to stochastic equilibria. As an illustrative example, consider stochastic bubbles, which were introduced by \citet{Blanchard1979} in a reduced-form model and by \citet{Weil1987} in a general equilibrium model. This model is similar to those of \citet{Samuelson1958} and \citet{Tirole1985} except that agents rationally expect the bubble to burst (permanently become worthless) with some probability each period. Since this model reduces to the perfect foresight economy once the bubble asset becomes worthless, under the low interest condition, we can eliminate this equilibrium by dividend injection. Consequently, stochastic bubbles due to sunspots can be ruled out. This argument casts doubt on the plausibility of the stochastic bubble model, which is widely applied; see, for instance, \citet[\S4]{CaballeroKrishnamurthy2006}, \citet[\S4.2]{FarhiTirole2012}, \citet[\S II.C]{MartinVentura2012}, \citet{AokiNikolov2015}, \citet*{HiranoInabaYanagawa2015}, \citet{Clain-Chamosset-YvrardKamihigashi2017}, \citet*[\S III]{BiswasHansonPhan2020}, and \citet*{Guerron-QuintanaHiranoJinnai2023}, among others.\footnote{In an interesting recent paper, \citet{HanWang2025} consider an OLG model of multiple rational bubbles with heterogeneous confidence (bursting probability), with cryptocurrencies as an application in mind. They focus on the stationary equilibrium and show its uniqueness, but it is clear from the analogy of \S\ref{sec:example} that there is a continuum of nonstationary equilibria.}

\subsection{Public debt as private liquidity \texorpdfstring{\citep{Woodford1990}}{}}\label{subsec:application_Woodford}

\citet{Woodford1990} develops a model with infinitely lived agents in which agents have investment opportunities every other period, private borrowing is impossible, and the only means of saving is government bonds backed by taxes. In this setting, he showed that government bonds may crowd in capital investments rather than crowd them out. This model is a variant of \citet{Bewley1980}'s model with growth and production. It has subsequently been applied to pure bubble models by replacing government bonds with pure bubble assets, including \citet{Kocherlakota2009}, \citet{HiranoYanagawa2017}, \citet{KiyotakiMoore2019}, and \citet*{BrunnermeierMerkelSannikov2024}. Since the mathematical argument for equilibrium selection is similar to that in \S\ref{subsec:Kocherlakota} and \citet{Woodford1990}'s mechanism is explained in \citet[\S4.2]{HiranoToda2024JME}, we omit the details.

\subsection{Bubble as store of value under idiosyncratic risk}\label{subsec:application_storage}

There are several economic models in which a pure bubble asset circulates in the presence of uninsurable idiosyncratic risk. \citet{Kitagawa1994,Kitagawa2001} showed in a two-period OLG economy that fiat money can circulate as a safe asset when agents face idiosyncratic storage risk. \citet{AokiNakajimaNikolov2014} extend this model to an infinite-horizon setting with general shocks. \citet*{BrunnermeierMerkelSannikov2024} also model fiat money as self-insurance against idiosyncratic investment risk.

We briefly describe the model of \citet{AokiNakajimaNikolov2014}. There is a continuum of agents with logarithmic utility \eqref{eq:log_utility}. Investing $k_t$ units of good at time $t$ yields an output of $y_{t+1}=z_{t+1}k_t$ at time $t+1$, where the productivity $z_t>0$ is \iid across agents and time. The agents can also trade a risk-free asset in zero net supply. Let $w_t>0$ be the wealth of a typical agent at the beginning of time $t$. Given the risk-free rate $R$, due to logarithmic utility and \iid shocks, the optimal consumption rule is $c_t=(1-\beta)w_t$ and the optimal portfolio problem reduces to
\begin{equation}
    \max_{\eta\ge 0} \E[\log(z\eta+R(1-\eta))], \label{eq:opt_portfolio}
\end{equation}
where $\eta\ge 0$ is the fraction of wealth invested in the technology. Using these rules and aggregating individual wealth, we obtain the system of equations
\begin{equation*}
    G(R)=\beta \E[z\eta(R)+R(1-\eta(R))] \quad \text{and} \quad s(R)=1-\eta(R)
\end{equation*}
in \eqref{eq:system}, where $\eta(R)\ge 0$ is the (unique) maximizer of \eqref{eq:opt_portfolio}.

This model is a special case of \citet{Toda2014JET} with log utility, one technology, and \iid shocks. Because agents make the same portfolio choice due to homotheticity, a risk-free asset in zero net supply is useless for risk sharing \citep[Proposition 10]{Toda2014JET}. However, as \citet{AokiNakajimaNikolov2014} show, the introduction of a pure bubble asset may improve risk sharing and hence welfare because it is in positive net supply.

The fundamental equilibrium condition $s(R)=0$ implies $\eta(R)=1$. Taking the first-order condition of \eqref{eq:opt_portfolio} at $\eta=1$, the fundamental equilibrium interest rate satisfies
\begin{equation*}
    0=\E\left[\frac{z-R}{z}\right]\iff R_f=1/\E[1/z].
\end{equation*}
The low interest condition is therefore
\begin{equation}
    G(R_f)>R_f\iff \beta \E[z]>1/\E[1/z]\iff \beta \E[z]\E[1/z]>1, \label{eq:bubble_Aoki}
\end{equation}
which is exactly the bubble existence condition in \citet{AokiNakajimaNikolov2014}.\footnote{The productivity $z$ in our notation corresponds to $A\theta$ in \citet{AokiNakajimaNikolov2014}, where $A>0$ is a constant and they normalize $\E[\theta]=1$. Therefore, \eqref{eq:bubble_Aoki} is equivalent to $\E[\beta/\theta]>1$, which is their Equation (41).} Because $\eta=0$ whenever $R>\E[z]$ due to concavity, it follows that $G(R)=\beta R<R$ if $R>\E[z]$. Thus all assumptions of Theorem \ref{thm:exist} are satisfied, and there exists a bubbly equilibrium with $G(R_b)=R_b$, which is condition (39) in \citet{AokiNakajimaNikolov2014}.

\section{Conclusion}

If a model does not make determinate predictions, its value as a tool for economic analysis is severely limited. Although the equilibrium indeterminacy in pure bubble models has been recognized for decades, the literature has focused on only one of a continuum of bubbly equilibria corresponding to a saddle path or a steady state, with no scientific basis. We have challenged this long-standing unsolved problem in the literature and proposed a solution. Our proposed idea of dividend injection not only places the existing models of pure bubbles on a firmer footing but also enables robust applications to policy and quantitative analysis. Furthermore, our idea escapes the criticism that it is unrealistic to use pure bubble assets without dividends to describe real-world bubbles attached to land, housing, and stocks.

\appendix

\section{Proofs}\label{sec:proof}

\subsection{Proof of Proposition \ref{prop:Pareto_rank}}

The initial old's consumption $c_0^o=b+P_0$ is strictly increasing in $P_0$. Let us show that the equilibrium utility of generation $t$ is strictly increasing in $P_0$, which implies that $E(P_0')$ strictly Pareto dominates $E(P_0)$ if $P_0<P_0'$. Using $p_t=G^{-t}P_t$ and \eqref{eq:xt_diff}, the equilibrium utility of generation $t$ is
\begin{align*}
    &\log(aG^t-P_t)+\beta \log(bG^{t+1}+P_{t+1})\\
    &=t\log G+\log(a-p_t)+\beta(t+1)\log G+\beta\log\left(b+\frac{bp_t}{\beta a-(1+\beta)p_t}\right)\\
    &=(1+\beta)\log(a-p_t)-\beta\log(\beta a-(1+\beta)p_t)+\text{constant},
\end{align*}
where ``constant'' is a term independent of $p_t$. Let
\begin{equation*}
    f(p)=(1+\beta)\log(a-p)-\beta\log(\beta a-(1+\beta)p).
\end{equation*}
Then $0<p<\frac{\beta}{1+\beta}a$, and we have
\begin{equation*}
    f'(p)=-\frac{1+\beta}{a-p}+\frac{\beta(1+\beta)}{\beta a-(1+\beta)p}=\frac{(1+\beta)p}{(a-p)(\beta a-(1+\beta)p)}>0.
\end{equation*}
Therefore, the utility of generation $t$ is strictly increasing in $p_t$. Since $p_t$ is strictly increasing in $p_0=P_0$ by \eqref{eq:pt}, it follows that the utility of generation $t$ is strictly increasing in $P_0$.

In the bubbly equilibrium $E\left(\frac{\beta a-b}{1+\beta}\right)$, we have $R=G$. Its Pareto efficiency follows from detrending the economy and applying Proposition 5.6 of \citet{BalaskoShell1980}. \hfill \qed

\subsection{Proof of Proposition \ref{prop:determinacy_Samuelson}}

\ref{item:determinacy_Samuelson1} The proof follows from an application of the Bubble Necessity Theorem of \citet[Theorem 2]{HiranoToda2025JPE}. See also their Example 2.

\ref{item:determinacy_Samuelson2} Using the definition of $\phi$ in \eqref{eq:h} and $\xi_b^*$ in \eqref{eq:steady_states}, we obtain the Jacobian
\begin{equation*}
    D\phi(\xi_b^*)=\begin{bmatrix}
        \beta a/b & -G_d/G\\
        0 & G_d/G
    \end{bmatrix}.
\end{equation*}
Since by assumption $\beta a>b$ and $G_d<G$, the eigenvalues of $D\phi(\xi_b^*)$ are $\beta a/b>1$ and $G_d/G\in (0,1)$. Therefore, $\xi_b^*$ is a hyperbolic fixed point and the local stable manifold theorem \citep[Theorem 8.9]{TodaEME} implies that for any sufficiently small $\xi_{2,0}>0$ (hence sufficiently small $D_0>0$), there exists a unique orbit $\set{\xi_t}_{t=0}^\infty$ converging to $\xi_b^*$. Since the difference equation \eqref{eq:xt_autonomous} is autonomous (does not explicitly depend on time), the same argument applies if we start the economy from any $t=T$. Therefore, for sufficiently small $\xi_{2T}>0$ (hence sufficiently small $D_0(G_d/G)^T$), there exists a unique orbit $\set{\xi_t}_{t=T}^\infty$ converging to $\xi_b^*$. Since $G_d<G$, it follows that for large enough $T>0$, there exists a unique orbit $\set{\xi_t}_{t=T}^\infty$ converging to $\xi_b^*$. But clearly \eqref{eq:xt_autonomous} can be solved backwards because \eqref{eq:xt_autonomous2} is linear and \eqref{eq:xt_autonomous1} can be solved for $\xi_t$ as
\begin{equation*}
    \xi_{1t}=\beta a\left(\frac{b}{\xi_{1,t+1}+\frac{G_d}{G}\xi_{2t}}+1+\beta\right)^{-1}>0.
\end{equation*}
Therefore, by backward induction, there exists a unique orbit $\set{\xi_t}_{t=0}^\infty$ converging to $\xi_b^*$ and the equilibrium is locally determinate. \hfill \qed

\subsection{Proof of Theorem \ref{thm:determinacy_Tirole}}

The nonexistence of an equilibrium path converging to $\xi_f^*$ follows from the same argument as in the proof of Proposition \ref{prop:determinacy_Samuelson}. To prove the local determinacy of $\xi_b^*$, we apply the implicit function theorem and the local stable manifold theorem. We divide the rest of the proof in several steps.

\begin{step}
Calculation of the Jacobian.
\end{step}

Using the definition of $\Phi$ in \eqref{eq:Phi_Tirole} and $G=f'(k_b)$, we obtain
\begin{align*}
    D_\xi\Phi(\xi_b^*)&=\begin{bmatrix}
        0 & -1 & 0\\
        c_1^y-\omega' & 1 & 0\\
        0 & 0 & -G_d/G
    \end{bmatrix}, &
    D_\eta\Phi(\xi_b^*)&=\begin{bmatrix}
        -\frac{\xi_{2,b}^*}{G}f''(k_b) & 1 & 1\\
        G+c_2^y & 0 & 0\\
        0 & 0 & 1
    \end{bmatrix}.
\end{align*}
Suppose for the moment that
\begin{equation}
    G+c_2^y=G+\partial c^y/\partial k'>0. \label{eq:suff_cond}
\end{equation}
Since $D_\eta\Phi(\xi_b^*)$ is block upper triangular, we have $\det(D_\eta\Phi(\xi_b^*))=-G-c_2^y<0$ by \eqref{eq:suff_cond}. Therefore, applying the implicit function theorem, we can solve $\Phi(\xi,\eta)=0$ locally as $\eta=\phi(\xi)$. To simplify notation, let $p=\omega'-c_1^y$, $q=G_d/G$, $r=-\frac{\xi_{2,b}^*}{G}f''(k_b)$, and $s=G+c_2^y$. Then $s>0$ by \eqref{eq:suff_cond} and
\begin{equation*}
    D\phi(\xi_b^*)=-[D_\eta\Phi(\xi_b^*)]^{-1}D_\xi\Phi(\xi_b^*)=\begin{bmatrix}
        p/s & -1/s & 0\\
        -pr/s & 1+r/s & -q\\
        0 & 0 & q
    \end{bmatrix}.\footnote{We have calculated this matrix using \textsc{Mathematica} to avoid human errors.}
\end{equation*}

\begin{step}
If \eqref{eq:suff_cond} holds, then $p>0$, $q\in (0,1)$, and $s>0$.
\end{step}
Since $R_f<G_d<G$, we have $q\in (0,1)$. Since $f''<0$, we have $r>0$. Let us show $p>0$. Noting that $\omega(k)=f(k)-kf'(k)$, we have $\omega'(k)=-kf''(k)>0$. Therefore, it suffices to show $c_1^y<0$. Substituting the budget constraint into the utility function, the utility maximization reduces to
\begin{equation*}
    \max_{c^y}U(c^y,f'(k')(\omega(k)-c^y)).
\end{equation*}
The first-order condition is
\begin{equation}
    U_1-f'(k')U_2=0\iff M(c^y,f'(k')(\omega(k)-c^y))-f'(k')=0,\label{eq:foc_Tirole}
\end{equation}
where $M\coloneqq U_1/U_2$ is the marginal rate of substitution. Applying the implicit function theorem, we obtain
\begin{subequations}
\begin{align}
    \frac{\partial c^y}{\partial k}&=-\frac{M_2f'(k')\omega'(k)}{M_1-f'(k')M_2}, \label{eq:c1y}\\
    \frac{\partial c^y}{\partial k'}&=-\frac{(M_2(\omega(k)-c^y)-1)f''(k')}{M_1-f'(k')M_2}. \label{eq:c2y}
\end{align}
\end{subequations}
The quasi-concavity of $U$ implies $M_1<0$ and $M_2>0$. Since $f'>0$ and $\omega'>0$, we obtain
\begin{equation*}
	p=\omega'-c_1^y=\omega'\frac{-M_1}{f'M_2-M_1}>0.
\end{equation*}

\begin{step}
If \eqref{eq:suff_cond} holds, then $D\phi(\xi_b^*)$ has one eigenvalue in $(1,\infty)$ and two eigenvalues in $(0,1)$ and the bubbly equilibrium is locally determinate.
\end{step}

Since $D\phi(\xi_b^*)$ is block upper triangular, one eigenvalue is $q\in (0,1)$. The other two eigenvalues are those of the first $2\times 2$ block. Its characteristic \AAT{polynomial} is
\begin{align*}
    P(z)&\coloneqq \det\begin{bmatrix}
        z-p/s & 1/s \\ pr/s & z-1-r/s
    \end{bmatrix}\\
    &=(z-p/s)(z-1-r/s)-pr/s^2.
\end{align*}
Then $P(0)=p/s>0$ and $P(1)=-r/s<0$, so one eigenvalue is in $(0,1)$ and the other is in $(1,\infty)$.

Since $\xi_0=(\xi_{10},\xi_{20},\xi_{30})=(k_0,P_0/N_0,D_0/N_0)$, the number of predetermined variables ($\xi_{10}$ and $\xi_{30}$, so two) equals the number of eigenvalues of $D\phi(\xi_b^*)$ less than 1 in absolute value. Therefore, by the local stable manifold theorem, $\xi_b^*$ is locally determinate. (See the textbook treatment of \citet[\S8.7]{TodaEME} for more details on this argument.)

\begin{step}
\eqref{eq:suff_cond} and \eqref{eq:EIS_cond} are equivalent.
\end{step}

To complete the proof of Theorem \ref{thm:determinacy_Tirole}, it remains to show the equivalence of \eqref{eq:EIS_cond} and \eqref{eq:suff_cond}. Let $\sigma=1/\varepsilon$ and $M=U_1/U_2$. Setting $t=c_1/c_2$, it follows from the definition of the elasticity of intertemporal substitution \eqref{eq:EIS} that
\begin{equation}
    \sigma=-t\frac{\partial \log M(tc_2,c_2)}{\partial t}=-tc_2\frac{M_1}{M}=-\frac{c^y}{f'}M_1, \label{eq:sigma1}
\end{equation}
where we have used \eqref{eq:foc_Tirole}. Similarly,
\begin{equation}
    \sigma=-t\frac{\partial \log M(c_1,c_1/t)}{\partial t}=\frac{c_1}{t}\frac{M_2}{M}=(\omega-c^y)M_2. \label{eq:sigma2}
\end{equation}
Therefore, using \eqref{eq:sigma1} and \eqref{eq:sigma2} to eliminate $M_1,M_2$ from \eqref{eq:c2y} and the bubbly steady state condition $G=f'(k_b)$, we obtain
\begin{align*}
    \eqref{eq:suff_cond}&\iff G-\frac{(\sigma-1)f''}{-\sigma \frac{f'}{c^y}-f'\frac{\sigma}{\omega-c^y}}>0\\
    &\iff \sigma \frac{G^2}{f''}\left(\frac{1}{c^y}+\frac{1}{\omega-c^y}\right)+\sigma-1<0\\
    &\iff \varepsilon=\frac{1}{\sigma}>1+\frac{G^2}{f''}\frac{\omega}{c^y(\omega-c^y)}\iff \eqref{eq:EIS_cond}. \hfill \qed
\end{align*}

\subsection{Proof of Proposition \ref{prop:locdet}}

The nonexistence of an equilibrium converging to the fundamental balanced growth path follows by the same argument as in Proposition \ref{prop:determinacy_Samuelson}.

To prove the local determinacy of the bubbly steady state, define the auxiliary variables
\begin{equation*}
    \xi_t=(\xi_{1t},\xi_{2t})=\left(y_t,\frac{\psi\beta}{1-\beta+\psi\beta}DG^{-t-1}\right).
\end{equation*}
Then \eqref{eq:eqcond_y} can be rewritten as $\xi_{t+1}=\phi(\xi_t)$, where $\phi:\R_{++}^2\to \R^2$ is defined by
\begin{align*}
    \phi_1(\xi_1,\xi_2)&=\frac{1}{G}\frac{\psi\beta}{1-\beta+\psi\beta}\xi_1(f'(\xi_1)+1-\delta)+\xi_2,\\
    \phi_2(\xi_1,\xi_2)&=\frac{1}{G}\xi_2.
\end{align*}
Let $\xi^*=(y_b,0)$ be the bubbly steady state, where $y$ satisfies \eqref{eq:eq_y}. Using the steady state condition \eqref{eq:eq_y}, a straightforward calculation yields the Jacobian
\begin{align*}
    D\phi(\xi^*)&=\begin{bmatrix}
        \frac{1}{G}\frac{\psi\beta}{1-\beta+\psi\beta}(f'(y_b)+1-\delta+y_bf''(y_b)) & 1\\
        0 & 1/G
    \end{bmatrix}\\
    &=\begin{bmatrix}
        1+\frac{1}{G}\frac{\psi\beta}{1-\beta+\psi\beta}y_bf''(y_b) & 1\\
        0 & 1/G
    \end{bmatrix}.
\end{align*}
Therefore, the eigenvalues of $D\phi(\xi^*)$ are
\begin{equation}
    \lambda_1=1+\frac{1}{G}\frac{\psi\beta}{1-\beta+\psi\beta}y_bf''(y_b) \label{eq:lambda1}
\end{equation}
and $\lambda_2=1/G\in (0,1)$. Since the initial condition $(k_0,D)$ is exogenous, the steady state is locally determinate if $\abs{\lambda_1}<1$. Since $f''<0$, we have $\lambda_1<1$ by \eqref{eq:lambda1}. Therefore, it suffices to check $\lambda_1>-1$, which is equivalent to \eqref{eq:locdet}.

Finally, let us show the equivalence of \eqref{eq:locdet} and \eqref{eq:locdet2}. Setting $y=K/L$ and using $F_K=f'(y)$ and $F_L=f(y)-yf'(y)$ in \eqref{eq:ES}, a straightforward calculation yields
\begin{equation*}
    \frac{1}{\varepsilon}=-y\frac{ff''}{f'(f-yf')}.
\end{equation*}
Therefore
\begin{equation*}
    yf''(y)=-\frac{1}{\varepsilon}\frac{f'\omega}{yf'+\omega}=-\frac{1}{\varepsilon}\frac{f'}{1+yf'/\omega},
\end{equation*}
where we have used $\omega=f-yf'$. Substituting this expression into \eqref{eq:locdet} and using the steady state condition \eqref{eq:eq_y}, we obtain \eqref{eq:locdet2}. \hfill \qedsymbol

\printbibliography

\end{document}